\documentclass[11pt]{article}
\usepackage[utf8]{inputenc}

\usepackage{geometry}
\usepackage{amsmath,amssymb}
\usepackage{graphicx}
\usepackage{hyperref}

\usepackage{amsmath}
\usepackage{amsthm}
\usepackage{hyperref}
\usepackage{cleveref}
\usepackage{nicefrac}

\usepackage{physics}

\usepackage{algorithm}
\usepackage{algpseudocode}

\usepackage{}

\usepackage[suppress]{color-edits}
\addauthor[sz]{sz}{red}
\addauthor[lb]{lb}{blue}

\usepackage{upgreek}

\newcommand{\IGNORE}[1]{}

\let\alpha\upalpha

\newcommand{\MinShare}{\mathtt{MinShare}}
\newcommand{\MaxProd}{\mathtt{MaxProd}}

\newtheorem{theorem}{Theorem}
\newtheorem{lemma}{Lemma}
\newtheorem{proposition}{Proposition}

\newtheorem{remark}{Remark}
\newtheorem{assumption}{Assumption}

\newtheorem{definition}{Definition}

\newcommand{\set}[1]{ \{ #1 \} }

\title{Optimizing Contracts in Principal-Agent Team Production}
\author{Shiliang Zuo \\
University of Illinois Urbana-Champaign \\
\texttt{szuo3@illinois.edu}
}
\date{}

\begin{document}
\maketitle

\begin{abstract}
  %This work studies the repeated principal-agent problem from an online learning perspective. The principal's goal is to learn the optimal contract that maximizes her utility through repeated interactions, without prior knowledge of the agent's type (i.e., the agent's cost and production functions). 

We study a principal-agent team production model. The principal hires a team of agents to participate in a common production task. The exact effort of each agent is unobservable and unverifiable, but the total production outcome (e.g. the total revenue) can be observed. The principal incentivizes the agents to exert effort through contracts. Specifically, the principal promises that each agent receives a pre-specified amount of share of the total production output. The principal is interested in finding the optimal profit-sharing rule that maximizes her own utility. We identify a condition under which the principal's optimization problem can be reformulated as solving a family of convex programs, thereby showing the optimal contract can be found efficiently. 

%For hetergeneous agent types, we identify a set of continuity conditions on the principal's utility function, which allows us to directly reduce the problem to Lipschitz bandits. For identical agents, we give a polynomial sample complexity scheme to learn the optimal contract based on inverse game theory. For strategic non-myopic agents, we design a low strategic-regret mechanism. Also, we identify a connection between linear contracts and posted-price auctions, showing the two can in fact be reduced to one another, and we give a regret lower bound on learning the optimal linear contract via a direct reduction to a result to Kleinberg and Leighton. We further consider the problem of learning the optimal contract in a team production setting. We identify a condition under which the problem can be reformulated as solving a family of convex programs, thereby showing the optimal contract can be found efficiently. 

%For all our results, we show how suitable modeling choices with the appropriate economic intuitions can make the problem \emph{easy}. 

\end{abstract}

\section{Introduction}
In principal-agent problems with moral hazard, the principal must incentivize the agent to complete certain tasks, whilst the exact effort of the agent is unobservable and unverifiable (\cite{holmstrom1979moral}). The principal must incentivize the agent through contracts, which are essentially performance-based pay schemes. In many principal-agent problems, instead of contracting with a single agent, the principal must contract with a team of agents (\cite{holmstrom1982moral}). As an example, consider a firm which produces a certain product, e.g., cars. There are different people in charge of different processes, including research, engineering, marketing, quality assurance, etc. Everyone is essential in the process. What is the appropriate way for the principal to motivate this team of agents? 

Since the effort of each agent is unverifiable, the principal must post contracts, i.e., performance-based pay schemes, to incentivize the agents. However, instead of directly observing the performance of each agent, it is possible that the principal can only observe the ``aggregated performance". In the car production example, the principal can only observe the total revenue generated from the production process, rather than directly measure the performance of each agent. Therefore, the principal can only post contracts that depend on the total production rather than individual performances. Then, essentially, the principal must find a way to share the total production to motivate the team to exert effort and in turn maximize her own utility. 

Once the principal decides on a way to share the total revenue, how will the agents respond? Since the total production may depend on complex interactions of each agent's effort, each agent's response not only depends on his own share but also depends on the shares of other agents. Formally, the teams' response should form a Nash equilibrium, and could potentially depend on the contract in complex ways. In the following, we will walkthrough a specific example illustrating the problem and our main approach. 

\subsection{An Example}
Consider the following production setting with two agents. Each agent $i\in\set{1,2}$ can supply an effort $a_i\in[0,+\infty)$, which is also his private cost. The production function $f(a_1, a_2) = 3 {a_1^{1/3}a_2^{1/3}} $, i.e., a Cobb-Douglas production with decreasing return to scale. The principal can post linear contracts $\beta = (\beta_1, \beta_2)$. Under the contract $\beta$, the share for agent $i$ is $\beta_i$, and the principal receives share $(1 - \beta_1 - \beta_2)$.

Assuming the agents' are risk neutral, in the Nash equilibrium, each agent should maximize his expected utility. Expressing the agent's response $a_i$ as a function of the contract $\beta$, then
\[
a_i(\beta) = \arg\max_{a_i} \beta_i f(a_i,a_{-i}) - a_i. 
\]

Excluding the degenerate equilibrium $(0,0)$, it can be checked the Nash equilibrium is
\[
a_1 = \beta_1^2 \beta_2, a_2 = \beta_2^2 \beta_1. 
\]
Then the production, expressed as a function of $\beta$, is $f(a(\beta)) = 3 \beta_1 \beta_2$. Hence, the principal's expected utility when posting contract $\beta$ is her own share multiplied by the production outcome
\[
(1 - \beta_1 - \beta_2) (3\beta_1 \beta_2)
\]
One would hope that the objective is concave so that the principal can apply some convex optimization algorithm and maximize the objective. However, it can checked the above objective is not concave, therefore at least the principal can not naively apply some convex optimization algorithm and hope the global optimum can be found. An important observation is that though the utility is not concave, the production $f(a(\beta)) = 3\beta_1\beta_2$ is \emph{quasiconcave} with respect to $\beta$. 

How does this help us? Instead of optimizing the utility directly, consider the following program
\[
\min \beta_1 + \beta_2, \quad \text{s.t. } f(a(\beta)) = 3\beta_1 \beta_2 \ge k. 
\]
The meaning of the above program is clear: it is the minimum total share that must be distributed to the agents, subject to the constraint that the production must meet a certain threshold $k$. Since $f(a(\beta))$ is quasiconcave, the constraint forms a convex set, and therefore the above program is a convex program. 
%Hence, for any production level $k$, the principal can apply some convex optimization algorithm and find the minimum total share that must be distributed to achieve a production level meeting the threshold $k$. 
The principal can then find her optimal utility via a two-stage process. In the first stage, the principal finds the optimal utility at production level $k$ using the above convex program. In the second stage, the principal finds the $k$ achieving the overall maximum utility. Essentially, the principal's original optimization problem, which is non-convex, can be reformulated as solving a family of convex programs. 

Another way to reformulate the principal's problem is to maximize the production level while constraining the total share distributed. 
\[
\max f(a(\beta)) = 3 \beta_1 \beta_2, \quad \text{s.t.} \beta_1 + \beta_2 \le k. 
\]
The above program is a constrained quasi-convex program, and can also be solved efficiently. Again, the principal can solve his original problem via a two-stage process. In the first stage, the principal finds the optimal utility when the total distributed share is $k$. In the second stage, the principal finds the total distributed share $k$ which maximizes her utility. 

\subsection{Summary and Outline}
In the example given above, the two reformulations both relied on a key quantity of interest $f(a(\beta))$, which is the production output expressed as a function of the contracts. We will term this the ``induced production function". The two reformulations both relied on the fact that the induced production function is quasiconcave. Note however that this is in fact a highly non-trivial property and in general does not hold. After introducing the basic model in \Cref{sec:team}, we identify a technical condition that guarantees the quasiconcavity of the induced production function in \Cref{sec:condition}. Then in \Cref{sec:reformulation}, we show how the principal's optimization problem can be reformulated as solving a family of convex programs, and also briefly discuss the connection of the reformulations to the concept of the cost function and the indirect production function in economics. % Finally in \Cref{sec:exp}, we propose algorithmic implementations for solving the reformulated families of (quasi-)convex programs; the implementations are also tested on numerical experiments. 

\szdelete{
Things to cover:
\begin{itemize}
\item Polynomial sample complexity against a myopic agent
\item The SDFC condition makes the problem instance nice
\item Applying the SDFC condition to stochastic or adversarial agents
\item Lower bound for Linear contracts. 
\item Linear contracts with strategic agents
\item General contracts against strategic agents? 
\item Reverse Moral Hazard? 
\end{itemize}
}

\section{Related Work}
Contract theory has been an important topic in economics, dating at least back to the work by \cite{holmstrom1979moral}. More recently, there have been numerous works in the computer science community studying contract theory from a computational or learning perspective. For example, the works by \cite{ho2014adaptive, zuo2024new, collina2024repeated} studies the problem of regret-minimization in repeated principal-agent problems; the works by \cite{dutting2019simple, dutting2021complexity, dutting2022combinatorial, dutting2023multi, duetting2024multi} study the algorithmic tractability of principal-agent problems of a combinatorial nature. 

The team production model with moral hazard was first introduced in the seminal work by~\cite{holmstrom1982moral}. More recently, there are some works that study multi-agent contract design from a computational and algorithmic perspective, e.g., \cite{dutting2023multi}, \cite{castiglioni2023multi}, \cite{deo2024supermodular}, \cite{cacciamani2024multi}. The problem studied in this work continues this line of research. However, the flavor of this work is quite different from all these prior works. Specifically, all these prior works are of a combinatorial nature, where the action space of the agent is assumed finite. In this work, the action space of the agent is assumed to be continuous, and therefore the problem of finding the optimal contract becomes a continuous optimization problem. As such, from a technical perspective, the results in this work are incomparable to these prior works. The continuous effort space model follows more closely with the economics literature. As a result, the approach in this work captures more commonly used production functions in economics, such as the CES family. As a remark, the recent work by \cite{zuo2024new} study learning algorithms for contract design under the continuous effort space model, however their work is limited to the single-agent case. %The results presented here conceivably strengthen the claims in \cite{} that argue the continuous action space model is the more natural choice in modeling economic productions. 

% Some additional work that studies contract design from a machine learning perspective include. The work by \cite{zuo} (re)-introduced the one-dimensional effort model, which this work adopts in the team production model. 

% The problem studied in this work is also somewhat related to the literature on learning in Stackelberg games, e.g. \cite{roth2016watch, peng2019learning, Letchford2009LearningAA, dong2018strategic}. In Stackelberg games, the leader commits to a strategy to which the agent best responds. In this line of work, the leader's optimization problem typically consists of two stages, the first stage involves inducing a particular action of the follower, and the second stage involves finding the best action to induce (e.g. \cite{roth2016watch}). 

%By contrast, in the principal-agent problem with moral hazard studied in this work, the agents' actions cannot be observed, and hence these typical approaches from Stackelberg games cannot be applied. Though our reformulation also consists of a two-stage process, it is fundamentally different from the literature on optimization in Stackelberg games. 

% \input{sec/hardness}

% \input{sec/condition_new}

\section{Team Production Model}
\label{sec:team}
The principal contracts with a group of $n$ agents. When each agent $i$ takes an action $a_i \in [0, \infty)$, the production is $f(a)$, where $a$ denotes vector $a = (a_1,\dots, a_n)$.  %Also assume the cost function of each agent $c(a) = a$. Note that this is without loss of generality, since for any increasing and continuous cost function $c$ with $c(0) = 0$, one can perform an appropriate reparameterization so that $c(a) = a$. %  Each agent has a cost function $c_i(\cdot)$. $c_i$ is assumed to be continuous, convex, and strictly increasing with $c_i(0)=0$. Note that without loss of generality, we can always reparametrize so that $c_i(a_i) = a_i$

The principal can write linear contracts specified by a $n$-tuple $\beta = (\beta_1, \dots, \beta_n)$. Here, the $i$-th component $\beta_i$ denotes agent $i$'s share under contract profile $\beta$. Assuming the agents are risk-neutral, the utility of agent $i$ under contract profile $\beta$ and the teams action $a$ is then:
\[
\beta_i f(a) - a_i. 
\]
The agents will respond with the Nash equilibrium action profile $a = (a_1, \dots, a_n)$ that satisfies the following:
\[
a_i \in \arg\max_{a'_i} \beta_i f(a_{-i}, a'_i) - a'_i. 
\]

We restrict attention to the equilibrium that satisfies the first-order conditions $\beta_i \partial_i f(a) = 1, \forall i \in [n]$. To do so, we impose the Inada conditions, which is a relatively standard condition. 

\begin{definition}
The production function $f(a) : (\mathbb{R^+})^n \rightarrow \mathbb{R}$ is said to satisfy the Inada condition if the following holds. 
\begin{enumerate}
\item $f$ is concave on its domain. 
\item $\lim_{a_i \rightarrow 0} \partial f(x) / \partial a_i = +\infty$. 
\item $\lim_{a_i \rightarrow +\infty} \partial f(x) / \partial a_i = 0$. 
\end{enumerate}
\end{definition}
% \begin{assumption}
% $f$ is strictly concave, continuously differentiable, and satisfies the Inada conditions. 
% \end{assumption}

\begin{proposition}
\label{prop:nasheq}
For any $\beta$, there exists a unique equilibrium satisfying the first-order conditions:
\begin{align*}
\beta_i \partial_i f(a) = 1. 
\end{align*}
\end{proposition}

\begin{definition}
Given a contract $\beta$, define the induced production function $F(\beta) := f(a(\beta))$ as the production $f(a)$ when $a$ is the unique equilibrium satisfying the first-order conditions given contract $\beta$. 
\end{definition}

The induced production function should be viewed as a function of $\beta$. By itself, the above definition is not too surprising. However, as the following will show, this quantity will be the key in analyzing the principal's optimization problem. When the contract is $\beta$, the principal's utility is then
\[
(1 - \sum_i \beta_i) f(a(\beta)). 
\]

\section{Quasiconcavity of the Induced Production Function}
\label{sec:condition}
%As we have mentioned in the introduction, the utility of the principal is in general not concave as a function of the contract profile $\beta$. 
% In this section, I first show a technical condition under which the production $f$ is quasiconcave with respect to the contract $\beta$. Then, I show how the principal's optimization problem can be formulated as solving a family of convex or quasiconvex programs. 

The main goal of this subsection is to provide a sufficient condition (namely, \Cref{assump:AdditiveForQuasiConcave} below) for which the induced production function $f(a(\beta))$ is quasiconcave.

\begin{assumption}
\label{assump:AdditiveForQuasiConcave} The production function $f(a)$ satisfies the following. 
\begin{enumerate}
\item $f$ is strongly separable so that there exists functions $h, g_{1:n}$ such that 
\[
f(a) = h(\sum_{i\in[n]} g_i(a_i)).
\]
Here $h$ is a monotonic transformation, meaning
\[
\sum_{i\in[n]} g_i(a_i) > \sum_{i\in[n]} g_i(b_i) \Rightarrow h(\sum_{i\in[n]} g_i(a_i)) > h(\sum_{i\in[n]} g_i(b_i)). 
\]
\item For each $g_i$, the function
\[
y_i(\cdot) = g_i\circ (1/g'_i)^{-1} (\cdot)
\]
is well-defined, strictly increasing, and concave. 
\end{enumerate}
% $f$ is strongly separable so that there exists functions $h, g_{1:n}$ such that $f(a) = h(\sum_{i\in[n]} g_i(a_i))$, 
% where $h$ is a monotonic transformation, i.e., 
% \[
% \sum_{i\in[n]} g_i(a_i) > \sum_{i\in[n]} g_i(b_i) \Rightarrow h(\sum_{i\in[n]} g_i(a_i)) > h(\sum_{i\in[n]} g_i(b_i)),
% \]
% and 
\end{assumption}

In the following, denote $g(\beta) = \sum_{i\in[n]} g_i(a_i(\beta))$. Then the production $f$, as a function of $\beta$, can be expressed as $f(a(\beta)) = h(g(\beta))$. The idea will be to show that $g(\beta)$ is quasiconcave, in particular, its upper-level sets are convex.

\szcomment{Fix typo: f' should be h' in the following}
\begin{lemma}
$g(\beta) = \sum_{i\in[n]} y_i (\beta_i h'(g(\beta)))$. 
\end{lemma}
\proof
In this proof denote $a = a(\beta)$. Since $a$ forms an equilibrium and the first-order conditions are met: 
\[
\beta_i = \frac{1}{h'( \sum_i g_i(a_i) )} \cdot \frac{1}{g'_i(a_i)}. 
\]
Substituting $g(\beta) = \sum_i  g_i(a_i)$ into the above expression and rearranging terms: 
\[
1 / g_i'(a_i) = \beta_i h'(g(\beta))
\]
By definition of $y_i$ and by applying the transformation $y_i$ to both sides in the above expression, we arrive at:
\[
g_i(a_i) = y_i (1/g'_i(a_i)) = y_i(\beta_i h'(g(\beta))). 
\]

Then for any $\beta$, 
\[
g(\beta)  = \sum_i y_i(\beta_i h'( g(\beta) )). \qedhere
\]
\endproof

% Define $g(\beta) = \sum g_1(a_1)$. Then the following holds:
% \[
% g(\beta) = \sum y_i (\beta_i z(g(\beta))). 
% \]
% So $g(\beta)$ is the value $t$ that solves
% \[
% t - \sum y_i (\beta_i z(t)) = 0
% \]
%Further, there is a unique solution to this problem. 

Fix any $t_0$. The upper-level sets $\set{\beta: g(\beta) \ge t_0}$ are still hard to analyze even with the above lemma since the right-hand side also involves the expression $g(\beta)$. We will relate this set to another set that `removes' $g(\beta)$ from the expression. 

\begin{lemma}
Fix any $t_0$. Denote $S_1 := \set{\beta : g(\beta) \ge t_0}$, $S_2 = \set{\beta : t_0 \le \sum_i y_i(\beta_i h'(t_0))}$. Then $S_1 = S_2$. 
\end{lemma}

\begin{proof}
We first show $S_2 \subset S_1$. Take any element $\beta\in S_2$, then we need to show $g(\beta) \ge t_0$. For sake of contradiction assume $g(\beta) = t_1 < t_0$. Then
\[
t_1 = \sum y_i (\beta_i h'(t_1)). 
\]
However, there must exists a $\zeta = (\zeta_1, \dots, \zeta_n) > \beta$, such that $g(\zeta) = t_0$. Then
\[
t_0 = \sum y_i( \zeta_i h'(t_0) ) > \sum y_i( \beta_i h'(t_0) ),
\]
which is a contradiction with the fact that $\beta \in S_2$. 

We next show $S_1 \subset S_2$. Take any $\beta \in S_1$, then we need to show
\[
t_0 \le \sum y_i(\beta_i h'(t_0)). 
\]
Since $t_0 \le g(\beta)$, there must exists a contract $\zeta = (\zeta_1, \dots, \zeta_n)$ such that $\zeta < \beta$ and that $g(\zeta) = t_0$, therefore
\begin{align*}
t_0 = \sum_i y_i(\zeta_i h'(t_0)) \le \sum_i y_i(\beta_i h'(t_0)). 
\end{align*}
Therefore $S_1 = S_2$. 
\end{proof}

We next show convexity of the set $S_2$. 
\begin{lemma}
The set $S_2$ (as defined in the previous lemma) is convex. 
\end{lemma}
\begin{proof}
The set $S_2$ is defined by
\[
\set{\beta: t_0 \le \sum_i y_i(\beta_i h'(t_0))}. 
\]
Keeping $t_0$ fixed, notice that both $t_0$ and $h'(t_0)$ are constants in the above expression. Further, the function $y_i$ is concave. This implies the convexity of the set $S_2$. 
\end{proof}

\begin{theorem}
The function $F(\beta) = f(a(\beta))$ is quasiconcave. 
\end{theorem}
\begin{proof}
Note $f(a(\beta)) = h(g(\beta))$. We have shown that $g$ is quasiconcave. $f$ is then a composition of an increasing function with a quasiconcave function, therefore it is quasiconcave. 
\end{proof}

%Note \Cref{assump:AdditiveForQuasiConcave} is satisfied for many commonly used production functions, including the CES family when the substitution parameter is negative. 

\paragraph{Examples} \Cref{assump:AdditiveForQuasiConcave} is satisfied for the CES production family when the substitution parameter is non-positive. %Note that non-positive substitution parameter ensures that every agent is `essential', i.e., $f(a) = 0$ if $a_i = 0$ for some $i$. 
\begin{enumerate}
\item Consider the CES production function with a negative substitution parameter. I.e., 
\[
f(a) = (\sum_i k_i a_i^{r})^{\nicefrac{d}{r}}. 
\]
Here $r < 0$ is the substitution parameter, and $d < 1$ is the return to scale. 
Then we can define $h(x) = (-x)^{\nicefrac{d}{r}}$ and $g_i(a_i) = -k_i a_i^{r}$. Then $f(a) = h(\sum g_i(a_i))$. 
%$f(a) = h( -x^{-t} + \dots )$, $g(x) = -x^{-t}$. Then $f(a) = h(\sum g_i(a_i))$. %Then
% \[
% 1/g' = x^{1+t}/ t. 
% \]
% \[
% (1/g')^{-1} = x^{1/(1+t)}. 
% \]
% \[
% y_i = g(x^{1/(1+t)}) = - x^{-t/(1+t)}
% \]
% is concave. 
\item Consider the Cobb-Douglas production function: $f(a) = \prod_i a_i^{k_i}$, which is a special case of the CES production function as the substitution parameter approaches 0. Taking $h(x) = \exp(x), g_i(a_i) = k_i \ln a_i$, then $f(a) = h(\sum_i g_i(a_i))$. 
\end{enumerate}
In both cases, it can be verified that \Cref{assump:AdditiveForQuasiConcave} holds. 

\medskip
In fact, for the CES production functions, one can obtain a closed-form expression for the induced production function $f(a(\beta))$, which also takes the CES form. For the CES production with negative substitution parameter, the following holds. 
\begin{proposition}
\label{prop:CESform}
Assume $f(a) = (\sum_{i} k_i a_i^r )^{d/r}$. Then the induced production function 
\[
f(a(\beta)) = \left[\sum_i (k_i \beta^r d^r)^{1/(1-r)} \right]^{\frac{r-1}{r} \cdot\frac{d}{d-1}}. 
\]
\end{proposition}

For the Cobb-Douglass production, the following holds. 
\begin{proposition}
\label{prop:CDform}
Assume $f(a) = (\prod_i a_i^{k_i})$. Then the induced production function
\[
f(a(\beta)) = \left[ \prod_i (k_i \beta_i)^{k_i} \right]^{1 / (1 - \sum_i k_i)}. 
\]
\end{proposition}
In fact, for the Cobb-Douglass production, the optimal contract has a very simple closed-form solution. 
\begin{proposition}
\label{prop:CDoptform}
Assume $f(a) = \left( \Pi_i a_i^{k_i} \right)$. The optimal contract is $\beta_i = k_i$. 
\end{proposition}

\section{Finding the Optimal Contract via Reformulation to Convex Programs}
\label{sec:reformulation}
Given the results in the previous section and assuming the induced production function is quasi-concave, the principal's problem can be reformulated as solving a family of convex or quasiconvex programs. 

\subsection{Production-Constrained Convex Program}
Consider the below program. 
\[
\min \sum_i \beta_i, \quad \text{s.t. } f(a(\beta)) \ge k. 
\]
This program returns the minimum possible total share that must be distributed when the production output is required to meet a certain threshold $k$. Further, since $f$ is quasiconcave, the constraint forms a convex set, hence it is a convex program. Denoting the objective as $\MinShare(k)$, the following result should be immediate. 

\begin{proposition}
The optimal utility for the principal is equal to $\sup_{k\in [0, \infty)} (1 - \MinShare(k)) \cdot k$. 
\end{proposition}

\begin{remark}
Assuming the principal has a first-order oracle to $f(a(\beta))$, the principal essentially has a separation oracle for the constraint set $f(a(\beta)) \ge k$. Hence, the above program can be then solved via algorithms that only require separation oracles, such as the ellipsoid method (see e.g. \cite{grotschel2012geometric}). 
\end{remark}

\begin{remark}
The above program bear some resemblence to the concept of the cost function in the theory of production in economics. Indeed, if the effort of each agent were observable and verifiable, the first-best solution is to compensate each agent for exactly the amount of effort he spent. Then, the cost function is defined as the minimum cost (i.e., total compensation to the agents) when the production must meet a certain threshold $k$:
\[
\min \sum_i a_i, \quad \text{s.t. } f(a) \ge k. 
\]
However, in our setting with moral hazard, the principal can only induce agents' efforts indirectly through contracts. Therefore, the reformulation can be seen as the cost function in the ``second-best" setting. The ``second-best" setting refers to the situation where the agents' preference are not known and incentive compatibility constraints must be imposed. 
% Therefore, the principal can only work 

% so that the program arising from the induced production function is still fundamentally different from the concept of the `standard' cost function. Indeed, the reformulation can be seen as the cost function in a `second-best' setting. 
\end{remark}

\subsection{Share-Constrained Quasiconvex Program}
Consider the below quasiconvex program. 
\[
\max f(a(\beta)), \quad \text{s.t.} \sum_i \beta_i \le k. 
\]
The program rturns the maximum production that can be achieved when the total share distributed to the agent is limited to $k$. Denoting the objective as $\MaxProd(k)$, the following proposition should be immediate. 
\begin{proposition}
The optimal utility for the principal is equal to $\sup_{k\in [0, 1]} (1 - k) \cdot \MaxProd(k)$. 
\end{proposition}
\begin{remark}
In the above program, the objective is a quasi-convex function and the constraint is clearly a convex set. The principal can use algorithms for constrained quasi-convex programs, such as the projected (normalized) gradient descent method or the Frank-Wolfe method. These methods are originally proposed for convex optimization (for a textbook treatment see e.g. \cite{bubeck2015convex}), however their convergence has also been analyzed for quasiconvex optimization (e.g. \cite{hazan2015beyond}, \cite{lacoste2016convergence}). 
\end{remark}
\begin{remark}
The above program bear some resemblence to the concept of the \emph{indirect production function} in economics. Specifically, in the first-best solution, the indirect production function captures the optimal production possible when the budget is restricted to some quantity $k$:
\[
\max_a f(a), \quad \text{s.t. } \sum_i a_i \le k. 
\]
However, similar to the previous reformulation, in our setting with moral hazard, the principal cannot directly work with agents' response. Therefore, the reformulation here can be seen as the indirect production function in a ``second-best" setting. 
\end{remark}

\section{Discussion}
% Since to the best of the author's knowledge, this is the first work to design optimization and learning algorithms to find the optimal contracts in a team production setting, there are several directions which could be fruitful for future research. 

This work studied computationally efficient optimization algorithms in a principal-agent team production setting. The problem studied in this work is most related to the recent works \cite{dutting2023multi, duetting2024multi}. However, these prior works study algorithms in a purely combinatorial setting (i.e., the action space of the agent is binary or discrete, and the outcome space is also binary or discrete). By contrast, in this work, the effort space of each agent is assumed to be a continuous interval, which is quite different from these prior works.

We note that the model used in this work is more aligned with standard production models commonly studied in economics and is capable of capturing widely used production functions, such as those in the CES family. In this study, we introduce the concept of an `induced production function' and demonstrate that the principal's optimization problem can be reformulated as a family of convex programs. These reformulations share conceptual similarities with the established notions of cost functions and indirect production functions in economics. We recognize that there remain several avenues for further improvement and refinement of the current approach, which we outline below.

% In this sense, the model in this work follows more closely with \cite{zuo2024new}, which (re-)introduced the one-dimensional effort model. However, their work only studies the problem of contracting against single agents. 

% which is the `standard' model of team production in economics. As such, the model used in this work is able to capture more commonly used economic production functions, such as the CES production function. 

\paragraph{Milder Assumptions Guaranteeing Quasiconcavity}
This work identified a technical condition that guarantees the quasiconcavity of $f$ as a function of the contract. It would be interesting to see whether the given condition can be relaxed.

\paragraph{Other Contract Schemes} This work studied the use of linear contracts in a team production setting. There are other interesting contract schemes, for example, rank-order tournaments (\cite{lazear1981rank}), relative performance outcomes (\cite{holmstrom1982moral}), etc. It would be interesting to formulate tractable mathematical models for which the optimal contract can be computed efficiently for these contract schemes.

\newpage
%Bibliography
\bibliographystyle{apalike}  
\bibliography{references}  

\newpage
\appendix
\section{Missing Proofs}
\subsection{Uniqueness and existence of Nash equilibirum}
\proof [Proof of \Cref{prop:nasheq}]
Without loss of generality assume $\beta > 0$. We need to show the equation
\begin{align*}
\nabla f(a) = (1 / \beta_1, \dots, 1/\beta_n)
\end{align*}
has a unique solution. Uniqueness is implied by the strict concavity of $f$. Existence is implied the the Poincare-Miranda theorem with the Inada conditions. 
\endproof

\subsection{Form of induced production function in CES class}

% \begin{proposition}
% Assume $f(a) = (\prod_i a_i^{k_i})$. Then
% \[
% f(a) = \left[ \prod_i (\beta_i k_i)^{k_i} \right]^{1 / (1 - \sum_i k_i)}. 
% \]
% \end{proposition}

% \begin{proposition}
% Assume $f(a) = (\sum_{i} k_i a_i^r )^{d/r}$. Then 
% \[
% f(a(\beta)) = t^{d/r} = \left[\sum_i (k_i \beta^r d^r)^{1/(1-r)} \right]^{\frac{r-1}{r} \cdot\frac{d}{d-1}}. 
% \]
% \end{proposition}
\begin{proof}[Proof of \Cref{prop:CESform}]
In this proof, as a shorthand, denote $t = \sum_i k_i a_i^r$. 
\begin{align*}
\partial_{a_i} f(a) &= \frac{d}{r}\cdot (\sum_i k_i a_i^r)^{d/r - 1} \cdot r \cdot (k_i a_i^{r-1}) \\
&= d\cdot t^{d/r - 1} k_i a_i^{r-1}
\end{align*}
Therefore,
\begin{align*}
a_i^{r-1} &= (d k_i \beta_i)^{-1} t^{1 - d/r}  \\
\Rightarrow k_i a_i^{r} &= k_i^{1 / (1-r)} (\beta d)^{(r/(1-r))} t^{(r-d) / (r-1)}
\end{align*}
Summing over all $i$:
\begin{align*}
t = t^{(r -d) / (r-1)} \sum_i [k_i \beta^r d^r]^{1/(1-r)} 
\end{align*}
which leads to
\[
t^{(d-1) / (r-1)} = \sum_i [k_i \beta^r d^r]^{1/(1-r)} 
\]
So
\[
f(a(\beta)) = t^{d/r} = \left[\sum_i (k_i \beta^r d^r)^{1/(1-r)} \right]^{\frac{r-1}{r} \cdot\frac{d}{d-1}}. 
\]
\end{proof}

\begin{proof} [Proof of \Cref{prop:CDform}]
By the first-order conditions,
\begin{align*}
\beta_i f(a)\cdot \frac{k_i}{a_i} = 1
\end{align*}
Rewriting the above
\[
a_i = \beta_i k_i f(a)
\]
Taking both sides to the $k_i$-th power and taking the product over all $i$, we arrive at
\[
f(a) = \prod a_i^{k_i} = \prod_i (\beta_i k_i f(a))^{k_i}
\]
Rearranging terms
\[
f(a)^{1 - \sum_i {k_i}} = \prod_i (\beta_i k_i)^{k_i}
\]
Therefore
\[
f(a) = \left[ \prod_i (\beta_i k_i)^{k_i} \right]^{1 / (1 - \sum_i k_i)}. 
\]
\end{proof}

\begin{proof} [Proof of \Cref{prop:CDoptform}]
The principal's program can then be written as
\begin{align*}
&\max_{\beta, a} (1 - \sum_i \beta_i) f(a) \\
\text{s.t. } & \beta_i \cdot \frac{k_i}{a_i} \cdot \prod_i{a_i}^{k_i} = 1
\end{align*}
Therefore 
\[
\beta_i = \frac{a_i}{k_i f(a)}. 
\]
Substituting in the objective function, the objective becomes
\[
\max_a f(a) - \sum_i \frac{a_i}{k_i}. 
\]
This is a concave function with respect to $a$. Taking the derivative with respect to $a_i$, the first-order condition is
\[
\frac{k_i}{a_i} \cdot \prod_i a_i^{k_i}  = \frac{1}{k_i}. 
\]
Substituting into the equilibrium's first-order conditions:
\[
\beta_i = k_i. \qedhere
\]
\end{proof}
% \input{sec/app_experiment}
% \input{checklist}

% Bibliography
% \bibliographystyle{plain}
% \bibliography{references}

\end{document}